\documentclass{amsart}

\usepackage{hyperref,amsmath,amssymb,xy,mathrsfs,amsthm,amsxtra,graphicx,verbatim}
\hypersetup{nesting=true,debug=true,naturalnames=true}
\usepackage[mathcal]{euscript}

\newcommand{\Z}{\mathbf{Z}}
\newcommand{\R}{\operatorname{\mathbf{R}}}
\newcommand{\Proj}{\operatorname{Proj}}
\newcommand{\sgn}{\mathsf{sgn}}

\theoremstyle{plain}
\newtheorem{thm}[equation]{Theorem}

\newtheorem{lem}[equation]{Lemma}

\newtheorem{prop}[equation]{Proposition}

\theoremstyle{remark}
\newtheorem{rmk}[equation]{Remark}

\newtheorem{exm}[equation]{Example}

\begin{document}

\title{The Gibbs phenomenon for the Krawtchouk polynomials}

\author{John Cullinan}
\address{Department of Mathematics, Bard College, Annandale-On-Hudson, NY 12504, USA}
\email{cullinan@bard.edu}
\urladdr{\url{http://faculty.bard.edu/cullinan/}}

\author{Elisabeth Young}
\address{Department of Mathematics, Bard College, Annandale-On-Hudson, NY 12504, USA}
\email{ey0200@bard.edu}

\begin{abstract}
We study the Fourier approximation $\mathcal{F}_N$ of the sign function by the Krawtchouk polynomials.  We give numerical evidence that the Gibbs phenomenon of the approximation differs from the classical Gibbs constant; this is in contrast to other families of orthogonal polynomials.  We also show that the steepness $\mathcal{F}_N'(0)$ of the approximation is bounded by explicitly proving  $\lim_{N \to \infty}  \mathcal{F}_N'(0) = \log 4$.  This is also in contrast to approximations by classical orthogonal polynomials, where the steepness has been shown to be unbounded as the degree increases.  
\end{abstract}


\maketitle

\section{Introduction} \label{introduction}

\subsection{Motivation}

Let $I \subseteq \R$ be a real interval and $\lbrace p_n \rbrace_{n=0}^\infty$ a set of continuous functions on $I$ that are orthogonal with respect to an inner product $\langle~,~\rangle$.  If $f:I \to \R$ is a real-valued function, then the $N$th Fourier approximation $\mathcal{F}_N$ of $f$ is defined pointwise on $I$ by
\[
\mathcal{F}_N(x) = \sum_{n=0}^N \Proj_{p_n}(f)\, p_n(x),
\]
where $\Proj$ denotes orthogonal projection using $\langle ~,~ \rangle$.  The classical Gibbs phenomenon occurs in the approximation of discontinuous functions by trigonometric Fourier series.  We immediately consider a special case that is illustrative of this project.  Let $a$ be a positive real number and let $I = [-a,a]$.  Let $f = \sgn:I \to \R$ be the sign function.

\begin{exm} \label{trig_gibbs}
If $a = \pi$ and $p_n(x) = \sin(nx)$, then 
\[
\mathcal{F}_N(x) = \frac{4}{\pi} \sum_{n=0}^N \frac{1}{2n+1} \sin((2n+1)x).
\]
\end{exm}

\begin{exm} \label{legendre_gibbs}
If $a=1$ and $T_n$ is the $n$th Chebyshev polynomial, then by \cite[Equation (3.10)]{fay},
\[
\mathcal{F}_N(x) = \frac{4}{\pi} \sum_{n=0}^N \frac{(-1)^{n+1}}{2n-1}T_{2n-1}(x).
\]
Similarly, as shown in \cite[Equation (6.14)]{fay}, if $p_n$ is the $n$th Legendre polynomial then
\[
\mathcal{F}_N(x) = \sum_{n=1}^N \left(2-\frac{1}{2n} \right)p_{2n-2}(0)p_{2n-1}(x).
\]
\end{exm}

In the first example it is a standard exercise to show that as $N \to \infty$ and $x \to 0^+$, $\mathcal{F}_N(x)$ ``overshoots'' the value 1 by a constant $\gamma$, given explicitly by
\[
\gamma = \frac{2}{\pi} \int_0^\pi \frac{\sin t}{t} \,{\rm d}t \approx 1.179.
\]
One can ask whether or not there exists a similar Gibbs constant for the series of Example \ref{legendre_gibbs}.  In fact, the articles \cite{cs}, \cite{dh}, and \cite{kaber} show that if we approximate $\sgn$ (or a suitable variant if $I = \R$) by the Chebyshev, Gegenbauer, Hermite, Jacobi, Laguerre, or Legendre polynomials, then there is a Gibbs constant associated to each of these series and it is exactly $\gamma$. In each of the cases listed above, including the trigonometric case, it is also the case that the \emph{steepness} $\mathcal{F}_N'(0)$ of the approximation at 0 is unbounded as $N \to \infty$ (in \cite[\S3.2]{kaber} the author explicitly shows how the steepness varies with the parameter in the one-parameter family of Gegenbauer polynomials).

A common feature of each of the aforementioned Fourier series is that $\mathcal{F}'_N(x)$ can be summed by the Christoffel-Darboux identity \cite[Thm.~3.2.2]{szego}.  For example if $H_n(x)$ denotes the $n$th Hermite polynomial (defined on the real line $\R$), then 
\begin{align*}
\mathcal{F}_{2N+1}(x) = \sum_{n=0}^N \left( \frac{1}{2n+1} \right) \, \frac{H_{2n}(0) H_{2n+1}(x)}{\|H_{2n}\|^2}. 
\end{align*}
Using the derivative rule $H_n'(x) = 2nH_{n-1}(x)$, we have 
\begin{align*}
\mathcal{F}_{2N+1}'(x) &= 2 \sum_{n=0}^N  \frac{H_{2n}(0) H_{2n}(x)}{\|H_{2n}\|^2} \\
&= \frac{(-1)^N}{2^{2N}N!\sqrt{\pi}} \frac{H_{2N+1}(x)}{x},
\end{align*}
where we have simplified the sum using the Christoffel-Darboux formula along with the known values of $H_{2n}(0)$; see \cite[\S 3]{cs} for a full derivation.  This last expression is useful from two points of view.  First, using known estimates on the zeroes of the Hermite polynomials \cite[(6.31.20), (6.31.23)]{szego}, one can approximate the smallest critical point $\theta_N$ of $H_{2N+1}(x)$ as 
\[
\frac{\pi}{\sqrt{2N+1}} < \theta_N < \frac{\pi}{\sqrt{2N+1}} \underbrace{\left\{ \frac{1}{2} + \frac{1}{2} \left[1 - \left( \frac{2\pi}{2N+1} \right)^2 \right]^{1/2}\right\}^{-1/2}}_{\to 1^+ \text{ as } N \to \infty}. 
\]
Second, by the fundamental theorem of calculus, we can write
\[
\mathcal{F}_N(\theta_N) = \frac{(-1)^N}{2^{2N}N!\sqrt{\pi}} \int_0^{\theta_N}  \frac{H_{2N+1}(x)}{x}\,{\rm d}x.
\]
It can then be shown, using standard approximations to the Hermite polynomials, that as $N \to \infty$ this value approaches $\gamma$. 

This raises the question of whether this Gibbs constant is ``universal'' in all continuous approximations of $\sgn$. In this paper we show that the {Krawtchouk polynomials} exhibit a Gibbs phenomenon, but the actual constant reflecting the overshoot is different. We also show that the steepness of the approximation is bounded as $N \to \infty$.

\subsection{Krawtchouk Polynomials} Fix a positive integer $N$ and let $p$ and $q$ be positive real numbers with $p+q = 1$.  On the interval $[0,N]$ define the set of \emph{Krawtchouk polynomials} $\lbrace K_n^{(p)} \rbrace_{n=0}^N$ by 
\[
K_n^{(p)}(x;N)	= \sum_{v=0}^n(-1)^{n-v}\binom{N-x}{n-v}\binom{x}{v}p^{n-v}q^v.
\]
Define the weighting function $w_N(x) = \binom{N}{x}p^xq^{N-x}$.  Then the $K_n^{(p)}(x;N)$ are orthogonal on $[0,N]$ with respect to the discrete inner-product 
\[
\langle f,g\rangle = \sum_{y=0}^N f(y)g(y) w_N(y).
\]
We now immediately specialize to case that we will study in this paper.

From now on we take $N$ to be even and set $p=1/2$; see Section \ref{other_p_section} for the case of $p \ne 1/2$.  We consider the shifted polynomials
\begin{align} \label{Krawtchouk_def}
k_n(x;N) := K_n^{(1/2)}(x+N/2;N) = \frac{1}{2^n} \sum_{v=0}^n(-1)^{n-v} \binom{N/2-x}{n-v}\binom{N/2+x}{v}
\end{align}
on the interval $I_N = [-N/2,N/2]$ with respect to the shifted inner product.  We study the Fourier expansion of $\sgn$ by the $k_n(x;N)$. 

A common feature of the classical orthogonal polynomials is that they satisfy a Sturm-Liouville differential equation (see \cite[\S 8.61]{szego}), whereas the Krawtchouk polynomials -- which are orthogonal on a discrete set -- do not.  It is precisely this differential equation which allows one to write the Fourier approximation in a form whose derivative can be summed by the Christoffel-Darboux formula.  The discrete nature of the Krawtchouk polynomials suggests that our methods will be based on combinatorics, not calculus.  Indeed, we will derive, by way of the Super Catalan Numbers, that
\[
\mathcal{F}_N(x) = \sum_{n=0}^{N-2} \frac{k_n(0;N)k_{n+1}(x;N)}{\|k_n(-;N)\|^2},
\]
which is is very similar to the form of the approximation for the classical polynomials.  See Section \ref{gibbs} for further discussion of Fourier expansions in other bases.

\subsection{Main Result}  Recall that the steepness of the Fourier approximation to $\sgn$  is $\mathcal{F}_N'(0)$ and that for the classical orthogonal polynomials, the steepness is unbounded as $N \to \infty$.  Our main result is that this is not the case for the Krawtchouk polynomials.

\begin{thm} \label{main_steepness_thm}
With all notation as above, we have
\[
\lim_{N\to \infty} \mathcal{F}_N'(0) = \log 4.
\]
\end{thm}

In Figure 1 we plot the degree-39 Krawtchouk approximation of  $\sgn$ (in green) versus the Fourier sine approximation (with maximum frequency 79).  The slope at the origin is 
\[
\mathcal{F}_{40}'(0) = \frac{3637485804655193}{2671465728531600} \approx 1.3616,
\]
while $\log 4 \approx 1.3862$.

\begin{figure}[h]
\includegraphics[height=5cm]{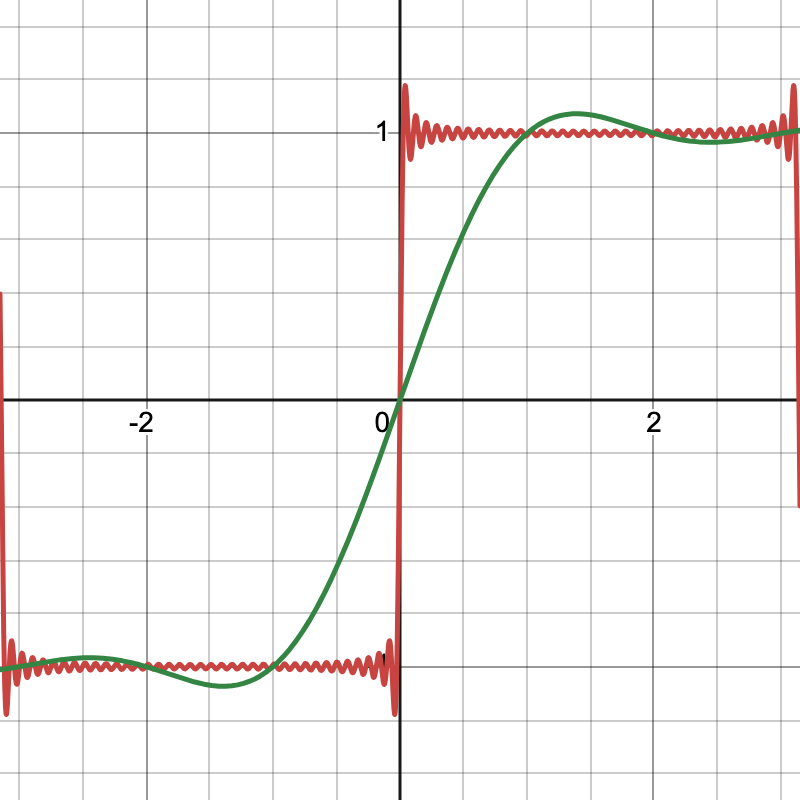} \label{picture}
\caption{Krawtchouk-vs-Trigonometric approximation}
\end{figure}

Computer calculations suggest that the overshoot remains bounded as well, \emph{i.e.},~that there is a Gibbs phenomenon associated to the Krawtchouk polynomials, but we are unable to supply a rigorous proof.  We explain the difficulty in proving such a result in Section \ref{gibbs}.  Instead, we give some numerical data here to support our assertion.  Let $\theta_N$ be the smallest positive critical point of $\mathcal{F}_{N}(x)$.  Then we find
\begin{center}
\begin{tabular}{rr}
$N$ & $\mathcal{F}_N(\theta_N)$ \\
\hline
10 & $\approx$ 1.101182 \\
50 & $\approx$ 1.071891 \\
100 & $\approx$ 1.068784 \\
200 & $\approx$ 1.067271 \\
400 & $\approx$ 1.066524 \\
500 & $\approx$ 1.066375 \\
600 & $\approx$ 1.066277
\end{tabular}
\end{center}

\subsection{Overview of the Paper} In the next section we prove that it is sufficient to study the Fourier approximations when $p=1/2$ by showing that $\mathcal{F}_N^{(p)}(x) = \mathcal{F}_N(x)$ for all $0 < p < 1$. Then we derive the  Krawtchouk Fourier  approximation to $\sgn$ and  show how it differs significantly from approximations in other orthogonal bases.  It is for this reason that we are unable to determine a closed form for $\lim_{N \to \infty} \mathcal{F}_N(\theta_N)$.  Then in Section \ref{steepness_section} we prove Theorem \ref{main_steepness_thm}.   All computer calculations were done using \texttt{Pari/GP} \cite{pari} using numerical precision to 500 digits.  For example, the approximation $\mathcal{F}_{600}(\theta_{600}) \approx 1.066277$ took 89 minutes on a standard laptop. We used the Desmos graphing calculator \cite{desmos} to produce Figure 1.

\subsection{Acknowledgment} The second author is grateful to the Bard College AIMS$^+$ program, in which this project was undertaken.

\section{Polynomial Interpolation and Independence of $p$}\label{other_p_section}

The function $\sgn: [-N/2,N/2] \to \R$ equips us with a set $\mathcal{P}_N$ of points in the plane:
\begin{align} \label{points}
\mathcal{P}_N &= \lbrace (-N/2,-1),\dots, (-1,-1), (0,0), (1,1), \dots, (N/2,1) \rbrace \nonumber \\&=  \lbrace (i-N/2,\sgn(i-N/2)) \rbrace_{i=0}^N.
\end{align}
Fix $p \in (0,1)$. The goal of this section is to prove that $\mathcal{F}_N^{(p)}(x)$ is independent of $p$.  We will do this by showing that $\mathcal{F}_N^{(p)}(x)$ is the (unique) interpolating polynomial of degree $N-1$ on the $N+1$ points of $\mathcal{P}_N$. Hence our convention to restrict our study to $p=1/2$. 

\begin{prop} \label{interp_prop}
With all notation as above, $\mathcal{F}_N^{(p)}(x)$ is a polynomial of degree $N-1$ that interpolates the points of $\mathcal{P}_N$.
\end{prop}

\begin{proof}
Since $\sgn$ is an odd function, the Fourier coefficients of the even-degree $k_n^{(p)}(x)$ are all zero.  Since $N$ is even, it follows that the maximal degree Krawtchouk polynomials appearing in $\mathcal{F}_N^{(p)}(x)$ is $N-1$.  Therefore, the degree of $\mathcal{F}_N^{(p)}(x)$ is $N-1$.  Since $\mathcal{F}_N^{(p)}(x)$ is an odd function we are reduced to showing that $\mathcal{F}_N^{(p)}(m) = 1$ for $m=1,\dots,N/2$.  

Fix $m \in \lbrace 1,\dots,N/2 \rbrace$.  We have
\begin{align*}
\mathcal{F}_N^{(p)}(m) &= \sum_{n=0}^N \frac{\langle \sgn,k^{(p)}_n(-;N) \rangle}{\| k^{(p)}_n(-;N) \|^2} k^{(p)}_n(m;N) \\
&= \sum_{n=0}^N \left( \sum_{y=-N/2}^{N/2} \sgn(y)k_n^{(p)}(y;N)w_N(y+N/2) \right) \frac{k_n^{(p)}(m;N)}{\| k^{(p)}_n(-;N) \|^2} \\
&= \sum_{y=-N/2}^{N/2} \sgn(y) w_N(y+N/2) \sum_{n=0}^N \frac{k_n^{(p)}(y;N)k_n^{(p)}(m;N)}{\| k^{(p)}_n(-;N) \|^2}.
\end{align*}
Write
\[
K_N(x,y) = \sum_{n=0}^N \frac{k_n^{(p)}(y;N)k_n^{(p)}(x;N)}{\| k^{(p)}_n(-;N) \|^2}.
\]
Then $K_N(x,y)$ is the Christoffel-Darboux kernel of the (shifted) Krawtchouck polynomials.  By \cite[Thm.~3.1.4]{szego}, we have
\[
K_N(x,y)w_N(y+N/2) = \delta_{xy}.
\]
Thus, 
\begin{align*}
\mathcal{F}_N^{(p)}(m)  &=  \sum_{y=-N/2}^{N/2} \sgn(y) w_N(y+N/2) K_N(m,y) \\
&= \sum_{y=-N/2}^{N/2} \sgn(y) \delta_{my} \\
&= \sgn(m) = 1.
\end{align*}
\end{proof}

As discussed above, Proposition \ref{interp_prop} immediately implies that $\mathcal{F}_N^{(p)}(x)$ is independent of $p$ because $\deg \mathcal{F}_N^{(p)} < |\mathcal{P}_N|$.  We will therefore set $p=1/2$ for the remainder of the paper and suppress the superscript ``$(1/2)$'' notation. To conclude this section we give an alternate form of $\mathcal{F}_N(x)$.  While this form is simpler to state, it is less useful for the purposes of Fourier analysis.  However, it will further highlight the deep connections to combinatorics exhibited by the Krawtchouk polynomials. 

Let $I_N(x)$ be the Lagrange form of the interpolating polynomial of the points of $\mathcal{P}_N$.  Then 
\[
I_N(x)=\sum_{i=0}^{N} \sgn(i-N/2)\ell_i(x),
\]
where
\[
\ell_i(x)=\prod_{j=0, j\neq i}^N \frac{x-(j-N/2)}{i-j}.
\]

\begin{prop} \label{interp_prop2}
With all notation as above, we have 
\[
I_N(x)=\sum_{i=0}^{N} \sgn(i-N/2)\binom{N/2+x}{i}\binom{N/2-x}{N-i}.
\]
\end{prop}

\begin{proof}
Since $I_N(x) = \sum_{i=0}^N\sgn(N/2-i)\ell_i(x)$, it remains to determine $\ell_i(x)$. From the definition:
\begin{align*}
\ell_i(x) &= \prod_{j=0, j\neq i}^N \frac{x-(j-N/2)}{i-j} \\
&=
\frac{(x+N/2)(x+N/2-1)\cdots(x+N/2-i+1)(x+N/2-i-1)\cdots(x-N/2)}{i!(N-i)!(-1)^{N-i}}\\
&=\frac{\binom{N/2+x}{i}i!\binom{N/2-x}{N-i}(N-i)!(-1)^{N-i}}{i!(N-i)!(-1)^{N-i}}\\
&=\binom{N/2+x}{i}\binom{N/2-x}{N-i}.
\end{align*}
This completes the proof.
\end{proof}

\begin{rmk}
Numerical experimentation suggests that the leading coefficient of ${I}_N(x)$ (and, therefore, $\mathcal{F}_N(x)$) is given by 
\[
\frac{(-1)^{(N-2)/2}}{(N-1)!} C\left( \frac{N-2}{2} \right),
\]
where $C(n)$ is the $n$th Catalan number.  See Appendix \ref{super_appendix} for further connections to the Catalan numbers. 
\end{rmk}

\section{The Approximation} \label{heuristic}

Recall from the Introduction that the polynomials $\lbrace k_n(x;N) \rbrace_{n=0}^N$ are orthogonal on $[-N/2,N/2]$ with respect to the shifted discrete inner-product
\begin{align*}
\langle f,g \rangle &= \sum_{y=-N/2}^{N/2} f(y)g(y) w_N(y+N/2) \\
&=\frac{1}{2^N}\sum_{y=-N/2}^{N/2} f(y)g(y) \binom{N}{y+N/2}.
\end{align*}
Taking $f = \sgn$, we have
\begin{align} \label{initial_Fourier_approx}
\mathcal{F}_N(x) = \sum_{n=0}^N \frac{\langle \sgn(x),k_n(x;N) \rangle}{\| k_n(x;N) \|^2} k_n(x;N).
\end{align}
A direct application of the inner-product shows that
\begin{align} \label{norm_identity}
\| k_n(x;N) \|^2 = 4^{-n}\binom{N}{n}.
\end{align}
Since $\sgn$ is odd, $w_N$ is even, and $k_n$ is an odd/even function in odd/even degree, it follows that
\begin{align} \label{fourier_projection}
\langle \sgn(x),k_n(x;N) \rangle = \begin{cases} \frac{2}{2^N} \sum_{y=1}^{N/2}k_n(y;N)\binom{N}{N/2+y} & \text{ if $n$ is odd, and} \\
0 & \text{ if $n$ is even.}
\end{cases} 
\end{align}

The remainder of this section is dedicated to determining an explicit expression for $\mathcal{F}_N(x)$ that mirrors the Fourier approximations of the sign function by other families of orthogonal polynomials.

\begin{prop} \label{fourier_approx}
With all notation as above, we have
\[
\mathcal{F}_N(x) = \frac{1}{2^{N-1}} \binom{N}{N/2} \sum_{n=0}^{N-2} \frac{k_{n}(0;N) k_{n+1}(x;N)}{\| k_{n}(-;N) \|^2}.
\]
\end{prop}

\begin{proof}
Applying Equation (\ref{fourier_projection}) to Equation (\ref{initial_Fourier_approx}) allows us to write
\begin{align}
\mathcal{F}_N(x) = \frac{1}{2^{N-1}} \sum_{n=0 \atop n \text{ odd}}^N \frac{1}{\| k_n(-;N)\|^2} \sum_{y=1}^{N/2} k_n(y;N)\binom{N}{N/2+y}k_n(x;N).
\end{align}
Now consider the product $k_n(y;N)\binom{N}{N/2+y}$:
\begin{align} \label{F12_identity}
k_n(y;N)\binom{N}{N/2+y} &= \frac{1}{2^n} \sum_{v=0}^{n} (-1)^{n-v} \binom{N/2-y}{n-v}\binom{N/2+y}{v}\binom{N}{N/2+y}.
\end{align}
Elementary manipulation of binomial coefficients allows us to write:
\begin{align*}
\binom{N/2+y}{v}\binom{N}{N/2+y} &= \binom{N}{v} \binom{N-v}{N/2-y},  \\
\binom{N/2-y}{n-v} \binom{N-v}{N/2-y} &= \binom{N-v}{n-v}\binom{N-n}{N/2+y-v}, \text{ and} \\
\binom{N-v}{n-v}\binom{N}{v} &= \binom{N}{n} \binom{n}{v}.
\end{align*}
Altogether, this implies
\begin{align*}
\mathcal{F}_N(x) &= \frac{1}{2^{N-1}} \sum_{n=0 \atop n\text{ odd}}^N \frac{\binom{N}{n}}{(-2)^n \|k_n(-;N)\|^2} \sum_{y=1}^{N/2} \sum_{v=0}^n (-1)^v \binom{n}{v} \binom{N-n}{N/2+y-v}k_n(x;N),
\end{align*}
which simplifies to
\begin{align*}
\mathcal{F}_N(x) &= \frac{1}{2^{N-1}} \sum_{n=0 \atop n\text{ odd}}^N (-2)^n \sum_{y=1}^{N/2} \sum_{v=0}^n (-1)^v \binom{n}{v} \binom{N-n}{N/2+y-v}k_n(x;N),
\end{align*}
using Equation (\ref{norm_identity}).  Interchanging the sums and reindexing, we rewrite $\mathcal{F}_N(x)$ as
\begin{align*}
\mathcal{F}_N(x) &= \frac{1}{2^{N-1}} \sum_{n=0 \atop n\text{ odd}}^N (-2)^n \left(\sum_{v=0}^{n} (-1)^v \binom{n}{v} \sum_{s=N/2+1-v}^{N-v} \binom{N-n}{s} \right)k_n(x;N).
\end{align*}

Recall that the Super Catalan numbers $S(n,m)$ are given by 
\[
S(n,m) = \frac{(2n)!(2m)!}{n!m!(m+n)!},
\]
and define $T(n,m) = \frac{2n+1}{m+n+1} S(n,m)$.  We point the reader to Appendix \ref{super_appendix} where we prove as part of a more general result that 
\begin{align} \label{super_identity_in_theorem}
\sum_{v=0}^n \sum_{s=N/2+1-v}^{N-v} (-1)^v\binom{n}{v}\binom{N-n}{s} = (-1)^{\frac{n+1}{2}}T\left( \frac{N-n-1}{2},\frac{n-1}{2}\right).
\end{align}

Applying Equation (\ref{super_identity_in_theorem}) to our expression for $\mathcal{F}_N(x)$ yields
\[
\mathcal{F}_N(x) = \frac{1}{2^{N-1}} \sum_{n=0 \atop n \text{ odd}}^{N}  (-2)^n(-1)^{\frac{n+1}{2}} \frac{(N-n)!(n-1)!}{\left( \frac{N-n-1}{2}\right)! \left(\frac{n-1}{2} \right)!} k_n(x;N).
\]
Finally, using the fact that 
\[
k_m(0;N) = \begin{cases} \binom{N/2}{m/2}\frac{(-1)^{m/2}}{2^{m}} & \text{ if $m$ is even, and} \\
0 & \text{ if $m$ is odd},
\end{cases} 
\]
together with Equation \ref{norm_identity}, allows us to reindex and simplify $\mathcal{F}_N(x)$ as
\[
\mathcal{F}_{N}(x) = \frac{1}{2^{N-1}}\binom{N}{N/2}\sum_{n=0}^{N-2} \frac{k_n(0;N)k_{n+1}(x;N)}{\| k_n(-;N)\|^2},
\]
as claimed.
\end{proof}

\section{Steepness} \label{steepness_section}

We now prove the main result of the paper, that the steepness of $\mathcal{F}_N(x)$ approaches $\log 4$ as $N \to \infty$.  As described in the Introduction, this is different from the classical orthogonal polynomials and trigonometric functions. 

\begin{thm} \label{steepness_thm}
The steepness of $\mathcal{F}_N(x)$ approaches $\log(4)$ as $N \to \infty$.
\end{thm}

We start by deriving an expression for $\mathcal{F}_N'(0)$.  By the linearity of the derivative, this will follow immediately once we have an expression for $k_m'(0;N)$, for positive integers $m$ and $N$ with $N$ even.

\begin{lem} \label{Krawtchouk_steepness}
Let $m,N \geq 1$ with $N$ even and $m \leq N$.  Then 
\[
k_m'(0;N) = \frac{1}{2^{m-1}}\sum_{\ell = 0}^{(m-1)/2} \frac{(-1)^\ell  \binom{N/2}{\ell}}{(m-2\ell)}
\]
if $m$ is odd, and 0 if $m$ is even.
\end{lem}

\begin{proof}
Since $k_m'(0;N)$ is the linear coefficient of $k_m(x;N)$, we require an expression for the linear term of Equation  (\ref{Krawtchouk_def}).  The generating function of $2^m k_m(x;N)$ is 
\[
G(t;N) = (1-t)^A(1+t)^B,
\]
where $A = N/2 - x$ and $B = N/2 + x$ \cite[(2.82.4)]{szego}. Then we may write
\begin{align*}
G(t;N) &= (1-t^2)^{N/2} \left( \frac{1+t}{1-t} \right)^x \\
&= (1-t^2)^{N/2} \exp \left(x \log \left( \frac{1+t}{1-t} \right) \right) \\
&= (1-t^2)^{N/2} \left(1 + x \log \left( \frac{1+t}{1-t} \right) + O(x^2)\right). 
\end{align*}
Therefore, the linear coefficient of $k_m(x;N)$ is the coefficient of $t^m$ in $(1-t^2)  \log \left( \frac{1+t}{1-t} \right)$, scaled by $2^{-m}$.  We multiply the power series representations
\[
(1-t^2)  \log \left( \frac{1+t}{1-t} \right) = \left( \sum_{\ell=0}^\infty (-1)^{\ell} \binom{N/2}{\ell} t^{2\ell} \right) \left( \sum_{\ell=0}^\infty \frac{t^{2\ell+1}}{2\ell+1} \right),
\]
and extract the coefficient of $t^m$ by convolution; we find it to be 
\[
2\sum_{\ell = 0}^{(m-1)/2} \frac{(-1)^\ell  \binom{N/2}{\ell}}{(m-2\ell)},
\]
if $m$ is odd and $0$ is $m$ is even.  Scaling by $2^{-m}$ completes the proof.
\end{proof}

\begin{prop} \label{exact_steepness_formula}
Let $c(n;N) = (-1)^{n/2} \binom{N/2}{n/2} / \binom{N}{n}$ if $n$ is even and 0 if $n$ is odd.  Then the steepness of  $\mathcal{F}_N(x)$ is given by
\[
\frac{1}{2^{N-1}} \binom{N}{N/2} \sum_{n=0}^{N-2} c(n;N) \sum_{\ell = 0}^{n/2} \frac{(-1)^\ell  \binom{N/2}{\ell}}{(n+1-2\ell)}.
\]
\end{prop}

\begin{proof}
Recall that 
\begin{align}
k_n(0;N) &= \begin{cases} \frac{(-1)^{n/2}}{2^n} \binom{N/2}{n/2} & \text{ if $n$ is even} \\
0& \text{ if $n$ is odd, and} 
\end{cases}  \label{k_n(0;N)}\\
\| k_n(-;N)\|^2 &= \frac{1}{4^n} \binom{N}{n}. \label{norm}
\end{align}
By Proposition \ref{fourier_approx}, we have
\[
\mathcal{F}_N(x) = \frac{1}{2^{N-1}} \binom{N}{N/2} \sum_{n=0}^{N-2} \frac{k_{n}(0;N) k_{n+1}(x;N)}{\| k_{n}(-;N) \|^2},
\]
hence
\[
\mathcal{F}_N'(0) = \frac{1}{2^{N-1}} \binom{N}{N/2} \sum_{n=0}^{N-2} \frac{k_{n}(0;N) k_{n+1}'(0;N)}{\| k_{n}(-;N) \|^2}.
\]
Now apply Equations (\ref{k_n(0;N)}) and (\ref{norm}) as well as Lemma \ref{Krawtchouk_steepness} to $k_{n+1}'(0;N)$.  Canceling the powers of $2^n$  yields the claimed formula for $\mathcal{F}_N'(0;N)$.
\end{proof}

Proposition \ref{exact_steepness_formula} equips us with an exact expression for $\mathcal{F}_N'(0)$, which we can evaluate for large $N$:
\begin{center}
\begin{tabular}{rr}
$N$ & $\mathcal{F}_N'(0)$ \\
\hline
400 & $\approx     1.38379$ \\
1000 & $\approx   1.38529$ \\
10000 & $\approx 1.38619$\\
20000 & $\approx 1.38624$
\end{tabular}
\end{center}
By comparison, we have $\log 4 \approx 1.38629$. 

\begin{rmk}
The explicit formula for $\mathcal{F}_N'(0)$ allows for much quicker computation than relying on the definition of $\mathcal{F}_N(x)$.  For example, the calculation of $\mathcal{F}_{20000}'(0)$ took roughly 2 minutes on a standard laptop, in contrast to those of Section \ref{introduction}.
\end{rmk}

Since the sum of Proposition \ref{exact_steepness_formula} is over even values of $n$, we re-index by setting $m = n/2$ and $M = N/2$ and define
\begin{align} \label{S(M)_definition}
S(M) = \frac{1}{2^{2M-1}}\binom{2M}{M} \sum_{m=0}^{M-1} (-1)^m \frac{\binom{M}{m}}{\binom{2M}{2m}} \sum_{\ell = 0}^{m} \frac{(-1)^{\ell}\binom{M}{\ell}}{2m+1-2\ell}.
\end{align}
We further define 
\begin{align} \label{T(M)_definition}
T(M) = 2 \sum_{k=1}^{M} \frac{1}{M+k},
\end{align}
and it is well known that $\lim_{M \to \infty} \sum_{k=1}^{M} \frac{1}{M+k} = \log 2$. We will prove that $S(M) = T(M)$ for all $M \geq 1$.  Then Theorem \ref{steepness_thm} will follow immediately as a corollary.

\begin{thm} \label{induction_theorem}
For $M \geq 1$ define $S(M)$ and $T(M)$ as in Equations (\ref{S(M)_definition}) and (\ref{T(M)_definition}).  Then $S(M) = T(M)$ for all $M \geq 1$.
\end{thm}

Before beginning the proof we set some convenient notation.  For integers $m$ and $M$ with $0 \leq m \leq M-1$, we define
\begin{align}
C(m,M) &=(-1)^m \frac{\binom{M}{m}}{\binom{2M}{2m}} \\
D(m,M) &= C(m,M) \left(\frac{2M-2m+1}{2M+2} \right). \label{D(m,M)_definition}
\end{align}
By applying well known properties of binomial coefficients, one verifies that
\begin{align}
\left( \frac{2M+1}{2M+2} \right) C(m,M+1) &= D(m,M),\text{ and} \label{CD_identity_1} \\
D(m,M) - D(m+1,M) &= C(m,M). \label{CD_identity_2} 
\end{align}
We also define 
\begin{align}
X(m,M) &= \sum_{\ell=0}^m \frac{(-1)^{\ell} \binom{M}{\ell}}{2m+1-2\ell}
\end{align}
and verify by Pascal's identity $\left(\binom{a}{b} = \binom{a-1}{b-1} + \binom{a-1}{b}\right)$ that
\begin{align} \label{X(m,M)_identity}
X(m,M+1) = X(m,M) - X(m-1,M).
\end{align}
We now prove an auxiliary lemma.

\begin{lem} \label{X(M,M+1)_lemma}
With all notation as above, we have 
\[
X(M,M+1) = (-1)^{M+1}\left(1 - \frac{2^{2M+1}}{\binom{2M+1}{M+1}}\right).
\]
\end{lem}

\begin{proof}
We have 
\[
X(M,M+1) = \sum_{\ell = 0}^{M} \frac{(-1)^{\ell} \binom{M+1}{\ell}}{2M+1-2\ell}.
\]
Write
\[
\frac{1}{2M + 1 -2\ell} = \int_0^1 x^{2M - 2\ell} \,{\rm d}x,
\]
so that 
\[
\sum_{\ell = 0}^{M} \frac{(-1)^{\ell} \binom{M+1}{\ell}}{2M+1-2\ell} = \int_0^1 x^{2M} \sum_{\ell = 0}^{M} \binom{M+1}{\ell} \left(-x^{-2}\right)^{\ell} \, {\rm d}x.
\]
By the Binomial Theorem 
\[
\sum_{\ell = 0}^{M+1} \binom{M+1}{\ell} \left(-x^{-2}\right)^{\ell} = (1-x^{-2})^{M+1},
\]
hence
\[
 \sum_{\ell = 0}^{M} \binom{M+1}{\ell} \left(-x^{-2}\right)^{\ell}  = (1-x^{-2})^{M+1} - (-x^{-2})^{M+1}.
\]
Therefore,
\begin{align*}
\sum_{\ell = 0}^{M} \frac{(-1)^{\ell} \binom{M+1}{\ell}}{2M+1-2\ell}  &= \int_0^1 x^{2M} \left((1-x^{-2})^{M+1} - (-x^{-2})^{M+1} \right) \,{\rm d}x \\
&=(-1)^{M+1} \int_0^1 \frac{(1-x^2)^{M+1} -1}{x^2}\,{\rm d}x.
\end{align*}
Integration-by-parts gives
\[
(-1)^{M+1} \int_0^1 \frac{(1-x^2)^{M+1} -1}{x^2}\,{\rm d}x = (-1)^{M+1} \left(1-2(M+1)\int_0^1(1-x^2)^{M}\,{\rm d}x\right).
\]
The latter is a Wallis integral with value
\[
\int_0^1 (1-x^2)^{M}\,{\rm d}x = \frac{2^{2M}}{(2M+1) \binom{2M}{M}}.
\]
It follows that 
\begin{align*}
\sum_{\ell = 0}^{M} \frac{(-1)^{\ell} \binom{M+1}{\ell}}{2M+1-2\ell} &= (-1)^{M+1} \left(1-2(M+1)\frac{2^{2M}}{(2M+1) \binom{2M}{M}}\right)\\
&= (-1)^{M+1}\left(1 - \frac{2^{2M+1}}{\binom{2M+1}{M+1}}\right),
\end{align*}
as claimed.
\end{proof}

With this notation in place we may write $S(M)$ more succinctly as
\begin{align} \label{new_S(M)_definition}
S(M) = \frac{1}{2^{2M-1}} \binom{2M}{M} \sum_{m=0}^{M-1} C(m,M)X(m,M).
\end{align}
We finish this section by proving Theorem \ref{induction_theorem} by induction.

\begin{proof}[Proof of Theorem \ref{induction_theorem}]
We verify $S(1) = T(1) =1$.  Fix $M \geq 1$, assume $S(M) = T(M)$, and consider $S(M+1)$:
\begin{align*}
S(M+1) &= \frac{1}{2^{2M+1}} \binom{2M+2}{M+1}\sum_{m=0}^M C(m,M+1)X(m,M+1) \\
&=  \frac{1}{2^{2M-1}} \binom{2M}{M}  \left(\frac{2M+1}{2M+2}\right) \sum_{m=0}^M C(m,M+1)X(m,M+1).
\end{align*}
Applying Equation (\ref{CD_identity_1}) we  rewrite $S(M+1)$ as
\begin{align}\label{S(M+1)_rewritten}
S(M+1) = \frac{1}{2^{2M-1}} \binom{2M}{M} \sum_{m=0}^{M} D(m,M)X(m,M+1).
\end{align}
Next, write
\[
\sum_{m=0}^{M} D(m,M)X(m,M+1) =\sum_{m=0}^{M-1} D(m,M)X(m,M+1) + D(M,M)X(M,M+1).
\]
By direct substitution in Equation (\ref{D(m,M)_definition}), we have 
\begin{align} \label{D(M,M)}
D(M,M) = \frac{(-1)^M}{2M+2}.
\end{align}
Applying Lemma \ref{X(M,M+1)_lemma}, we then have 
\begin{align} \label{D(M,M)X(M,M+1)}
\frac{\binom{2M}{M}}{2^{2M-1}}D(M,M)X(M,M+1) = \frac{2}{2M+1} - \frac{\binom{2M}{M}}{2^{2M-1}(2M+2)}.
\end{align}
Now use Equation (\ref{X(m,M)_identity}) to write
\begin{align*}
\sum_{m=0}^{M-1} D(m,M)X(m,M+1) &= \sum_{m=0}^{M-1} D(m,M)\left(X(m,M) - X(m-1,M)\right) \\
&=\sum_{m=0}^{M-1} D(m,M)X(m,M) - \sum_{m=0}^{M-1} D(m,M)X(m-1,M).
\end{align*}
Since $X(-1,M) = 0$ for all $M$, we reindex the second sum as
\begin{align} \label{reindex}
\sum_{m=0}^{M-1} D(m,M)X(m-1,M) = \sum_{m=0}^{M-2} D(m+1,M)X(m,M).
\end{align}
Now we write
\begin{align} \label{breakup}
\sum_{m=0}^{M-1} D(m,M)X(m,M) = \sum_{m=0}^{M-2} D(m,M)X(m,M) + D(M-1,M)X(M-1,M).
\end{align}
Applying Equations (\ref{reindex}) and (\ref{breakup}) we have
\begin{align*}
\sum_{m=0}^{M-1} D(m,M)X(m,M+1) &= \sum_{m=0}^{M-2}\left(D(m,M) - D(m+1,M)\right)X(m,M) \\
& + D(M-1,M)X(M-1,M).
\end{align*}
By Equation (\ref{CD_identity_2}), we have 
\begin{align} \label{condense}
\sum_{m=0}^{M-2}\left(D(m,M) - D(m+1,M)\right)X(m,M) = \sum_{m=0}^{M-2}C(m,M)X(m,M),
\end{align}
and by direct calculation (using Lemma \ref{X(M,M+1)_lemma}) we verify that 
\begin{align} \label{simplification}
\frac{\binom{2M}{M}}{2^{2M-1}}D(M-1,M)X(M-1,M) &= \frac{2}{2M-1} - \frac{2}{2M+2} \\
&- \frac{3\binom{2M}{M}}{2^{2M-1}(2M-1)(2M+2)}. \nonumber
\end{align}
Multiplying Equation (\ref{condense}) by $\frac{\binom{2M}{M}}{2^{2M-1}}$ and adding the right-hand-sides of Equations (\ref{D(M,M)X(M,M+1)}) and  (\ref{simplification}), we have
\begin{align*}
S(M+1) &= \frac{\binom{2M}{M}}{2^{2M-1}} \sum_{m=0}^{M-2}C(m,M)X(m,M)  - \frac{\binom{2M}{M}}{2^{2M-1}(2M+2)}\left(1 + \frac{3}{2M-1} \right) \\&+\frac{2}{2M-1} - \frac{2}{2M+2} + \frac{2}{2M+1}.
\end{align*}
Since 
\begin{align*}
\frac{\binom{2M}{M}}{2^{2M-1}(2M+2)}\left(1 + \frac{3}{2M-1} \right) &= \frac{\binom{2M}{M}}{2^{2M-1}(2M-1)}\\
\end{align*}
and
\begin{align*}
-\frac{\binom{2M}{M}}{2^{2M-1}(2M-1)} + \frac{2}{2M-1} = \frac{1}{2M-1} \left(\frac{2^{2M}}{\binom{2M}{M}} -1 \right) = C(M-1,M)X(M-1,M) 
\end{align*}
(using Lemma \ref{X(M,M+1)_lemma}, again), we have
\begin{align*}
S(M+1) &= S(M)  - \frac{1}{M+1} + \frac{2}{2M+1}.
\end{align*}
Applying the induction hypothesis $S(M) = T(M)$, we have
\begin{align*}
S(M+1) &= 2 \sum_{k=1}^M \frac{1}{M+k} - \frac{1}{M+1} + \frac{2}{2M+1} \\
&=2\sum_{k=1}^{M+1} \frac{1}{M+1+k} = T(M+1),
\end{align*}
which completes the proof.
\end{proof}

\section{The Gibbs Phenomenon and Connection to Classical Families} \label{gibbs}

Classical orthogonal polynomials typically satisfy a derivative rule (inherited from the Sturm-Liouville differential equation) wherein the derivative of a degree $n$ member is related to degree $n-1$ member. For example, the Jacobi, Generalized Laguerre, and Hermite polynomials satisfy
\begin{align*}
\frac{{\rm d}}{{\rm d}x} P_{n}^{(\alpha,\beta)}(x) &= \frac{\alpha+\beta+n+1}{2}P_{n-1}^{(\alpha + 1,\beta+1)}(x), \\
\frac{{\rm d}}{{\rm d}x} L_n^{(\alpha)}(x) &= -L_{n-1}^{(\alpha+1)}(x), \text{ and} \\
\frac{{\rm d}}{{\rm d}x} H_n(x) &= 2nH_{n-1}(x), 
\end{align*}
respectively.  The Gegenbauer and Legendre polynomials, both of which are subfamilies of the Jacobi polynomials, satisfy similar relations.  It is precisely these derivative formulas that allow one to express $\mathcal{F}_N'(x)$ as a Christoffel-Darboux sum (see Section \ref{introduction} for the Hermite, \cite[Prop.~1]{kaber} for the Gegenbauer, and \cite[\S 4]{cs} for the Generalized Laguerre).  This allows one to express the critical points of $\mathcal{F}_N(x)$ as the zeroes of a single classical orthogonal polynomial, which have been extensively studied for their asymptotic properties. By contrast, the Krawtchouck polynomials satisfy the difference equation
\[
k_n(x+1;N) - k_n(x;N) = k_{n-1}(x+1/2;N-1).
\]

Since $k'_{n+1}(x;N)$ is not a constant multiple of $k_{n}(x;N)$, it is not possible to directly apply the Christoffel-Darboux formula to 
\[
\mathcal{F}_N'(x) = \frac{1}{2^{N-1}} \binom{N}{N/2} \sum_{n=0}^{N-2} \frac{k_n(0;N)k_{n+1}'(x;N)}{\|k_n(-;N)\|^2}.
\]

In conclusion we mention a deep relationship between the Krawtchouk and Hermite polynomials, known to Krawtchouk \cite{krav} and explained in \cite{mina}.  With all notation as in Section \ref{introduction}, it is shown in \cite[Equation (4)]{mina} that for fixed $n$, 
\[
\lim_{N \to \infty} \left(\frac{2}{Npq}\right)^{n/2} n!K_n^{(p)}(\sqrt{2Npq}x + Np;N) = H_n(x),
\]
the $n$th Hermite polynomial.  However, the approximations for fixed $N$ only hold for $n = O(N^{1/3})$ and $x = O(n^{1/2})$.  Therefore, it is not feasible to use the known Gibbs constant for the Hermite polynomials to deduce an exact value for the Krawtchouk.

\appendix

\section{Super Catalan Numbers} \label{super_appendix}

Recall from \cite[\S 6]{gessel} that the \emph{Super Catalan Numbers} $S(n,m)$ are defined as
\[
S(p,q) = \frac{(2p)!(2q)!}{p!q!(p+q)!}.
\]
We further define 
\[
T(p,q) = \frac{(2p+1)!(2q)!}{p!q!(p+q+1)!} = \frac{2p+1}{p + q +1} S(p,q).
\]
We use this appendix to prove a lemma involving the Super Catalan Numbers that would have detracted from the main flow of the proof of Proposition \ref{fourier_approx}. 

\begin{lem} \label{appendix_lemma}
Let $N$ be even and $n$ be odd with $n \leq N$.  Then
\[
\sum_{v=0}^n \sum_{s=N/2+1-v}^{N-v} (-1)^v\binom{n}{v}\binom{N-n}{s} = (-1)^{\frac{n+1}{2}}T\left( \frac{N-n-1}{2},\frac{n-1}{2}\right).
\]
\end{lem}

In order to prove Lemma \ref{appendix_lemma} we first prove the following identity for $T(p,q)$.

\begin{prop} \label{T(p,q)_identity}
With all notation as above, we have 
\[
(-1)^{q+1} T(p,q) = \sum_{v=0}^{2q+1} \sum_{s=p+q+2-v}^{2p+2q+2-v} (-1)^v \binom{2q+1}{v}\binom{2p+1}{s}.
\]
\end{prop}

\begin{proof}
Consider the sum 
\begin{align} \label{app_sum_1}
 \sum_{v=0}^{2q+1} \sum_{s=p+q+2-v}^{2p+2q+2-v} (-1)^v \binom{2q+1}{v}\binom{2p+1}{s}
\end{align}
and set $n=2q+1$, $M = 2p+1$, $K=p+q+2$, and $N= 2p+2q+2$.  Then the sum in (\ref{app_sum_1}) can be written as
\begin{align} \label{app_sum_2}
\sum_{v=0}^n (-1)^v \binom{n}{v} \sum_{s=K-v}^{N-v} \binom{M}{s}.
\end{align}
Let $f(v) = \sum_{s=K-v}^{N-v} \binom{M}{s}$ so that 
\begin{align} \label{app_sum_3}
\sum_{v=0}^n (-1)^v \binom{n}{v} \sum_{s=K-v}^{N-v} \binom{M}{s} = \sum_{v=0}^n (-1)^v \binom{n}{v} f(v).
\end{align}

Recall that if $h: \Z \to \Z$ is a function, then the \emph{forward difference} $\Delta f:\Z \to \Z$ is defined as
\[
\Delta h(z) = h(z+1) - h(z).
\]
Iterating (using the linearity of  $\Delta$), it is straightforward to check that
\[
\Delta^{(n)} f(0) = \sum_{v=0}^n (-1)^{n-v} \binom{n}{v}f(v),
\]
whence
\begin{align} \label{app_sum_4}
\sum_{v=0}^n (-1)^v \binom{n}{v} f(v) = (-1)^n \Delta^{(n)} f(0).
\end{align}
We compute directly that
\[
\Delta f(v) = \binom{M}{k-v-1} - \binom{M}{N-v}.  
\]
Then 
\begin{align}
(-1)^n \Delta^{(n)} f(0) &= (-1)^n \Delta^{(n-1)} \Delta f(0)  \nonumber \\
&=  (-1)^n \sum_{v=0}^{n-1-v} \binom{n-1}{v} \Delta f(v) \nonumber \\
&= (-1)^n \sum_{v=0}^{n-1} (-1)^{n-1-v} \binom{n-1}{v} \left[  \binom{M}{K-v-1} - \binom{M}{N-v}\right]. \nonumber 
\end{align}
Substitute back in for $n$, $M$, $K$, and $N$:
\begin{align} 
(-1)^{2q+1} \sum_{v=0}^{2q} (-1)^{2q-v}\binom{2q}{v} \left[ \binom{2p+1}{p+q+1-v} - \binom{2p+1}{2p + 2q + 2 - v} \right]. \nonumber
\end{align}
Simplifying, and noting that $\binom{2p+1}{2p + 2q + 2 - v}  = 0$ in the range $v \in [0,2q]$, allows us to rewrite the sum in (\ref{app_sum_1}) as 
\begin{align} \label{app_sum_5}
- \sum_{v=0}^{2q} (-1)^v \binom{2q}{v} \binom{2p+1}{p+q+1-v}.
\end{align}
We have reduced the problem to proving that
\[
(-1)^qT(p,q) = \sum_{v=0}^{2q} (-1)^v \binom{2q}{v} \binom{2p+1}{p+q+1-v}.
\]
By the work of  \cite[\S 3]{gs}, the sum on the right is interpreted as the coefficient of $x^{p+q+1}$ in the product $(1-x)^{2q}(1+x)^{2p+1}$, which is shown in \emph{loc.~cit.}~to be 
\[
(-1)^q\frac{\binom{2p+1}{p}\binom{2q}{q}}{\binom{p+q+1}{q}}, 
\]
which is precisely $(-1)^q T(p,q)$.
\end{proof}

It is now a simple matter to complete the proof of Lemma \ref{appendix_lemma}. 

\begin{proof}[Proof of Lemma \ref{appendix_lemma}]
Set $p=(N-n-1)/2$ and $q=(n-1)/2$ in Proposition \ref{T(p,q)_identity}.
\end{proof}

\section{Declarations}

\subsection{Funding} The second-named author was partially supported by the Bard College AIMS$^+$ grant.

\subsection{Competing Interests} The authors have no competing interests to declare that are relevant to the content of this article.

\end{document}